\documentclass{article}
\usepackage[preprint]{neurips_2019}

\usepackage[utf8]{inputenc} % allow utf-8 input
\usepackage[T1]{fontenc}    % use 8-bit T1 fonts
\usepackage{hyperref}       % hyperlinks
\usepackage{url}            % simple URL typesetting
\usepackage{booktabs}       % professional-quality tables
\usepackage{amsfonts}       % blackboard math symbols
\usepackage{nicefrac}       % compact symbols for 1/2, etc.
\usepackage{microtype}      % microtypography
\pdfoutput=1

\usepackage{url}
\usepackage{amsmath}
 
\usepackage{amsthm}
\usepackage{algorithm}
\usepackage[noend]{algorithmic}
\usepackage{verbatim}
\usepackage{graphicx}
\usepackage[space]{grffile}
\usepackage{tikz}
\usepackage{subfigure}
\usetikzlibrary{arrows}
\newtheorem{theorem}{Theorem}

\newtheorem{lemma}{Lemma}

\newtheorem{proposition}{Proposition}

\theoremstyle{definition}

\graphicspath{{figures/}}
\usepackage{amsmath}
\usepackage{mathrsfs}
\usepackage{amssymb}
\usepackage{booktabs}
\usepackage{changepage}
\usepackage{xspace}
\usepackage{hyperref}

\newcommand{\epsi}[0]{ \varepsilon }
\newcommand{\reals}[0]{ \mathbb{R}^+ }

\newcommand{\oh}[1]{O\left( #1 \right)}
\newcommand{\ig}{InterlaceGreedy\xspace}
\newcommand{\fig}{\texttt{FIG}\xspace}

\newcommand{\add}{\texttt{ADD}\xspace}

\newcommand{\stopGain}{\epsi M/ n}

\DeclareMathOperator*{\argmax}{arg\,max}

\usepackage{tikz}
\usetikzlibrary{arrows}

\usepackage{xcolor}
\hypersetup{
    colorlinks,
    linkcolor={red!80!black},
    citecolor={blue!60!black},
    urlcolor={blue!80!black}
}

\begin{document}
\title{Interlaced Greedy Algorithm for Maximization of Submodular Functions in Nearly Linear Time}
%\title{Fast Algorithms for Maximizing (Non-Monotone) Submodular Functions Under Cardinality Constraint}
\author{Alan Kuhnle \\ Department of Computer Science \\ Florida State University \\ \texttt{akuhnle@fsu.edu}}

\maketitle
\begin{abstract}
A deterministic approximation algorithm is presented
for the maximization of non-monotone submodular functions over
a ground set of size $n$ subject to
cardinality constraint $k$; the algorithm is based upon the idea of
interlacing two greedy procedures. The algorithm
uses interlaced, thresholded
greedy procedures to obtain tight ratio $1/4 - \epsi$ in
$O \left( \frac{n}{\epsi} \log \left( \frac{k}{\epsi} \right) \right)$
queries of the objective function,
which improves upon both the ratio and 
the quadratic time complexity of the previously fastest deterministic
algorithm for this problem.
The algorithm is validated in the context of two applications
of non-monotone submodular maximization, on which it
outperforms the fastest deterministic and randomized
algorithms in prior literature.
\end{abstract}
\section{Introduction}
A nonnegative function $f$ defined on subsets of a ground set $U$
of size $n$
is \emph{submodular} iff 
for all $A,B \subseteq U$, $x \in U \setminus B$, 
such that $A \subseteq B$, it holds that 
$ f\left(B \cup x \right) - f(B) \le f\left(A \cup x \right) - f(A).$
Intuitively, the property of submodularity captures 
diminishing returns.
Because of a rich variety of applications, the maximization
of a nonnegative submodular % \footnote{For technical definitions of terms used in the Introduction, the reader is referred to Preliminaries below.}
function with respect to a
cardinality constraint (MCC)
has a long history of study \citep{Nemhauser1978}. 
Applications of MCC 
include viral marketing \citep{Kempe2003}, network monitoring
\citep{Leskovec2007}, %. Applications with $p$ matroid constraints 
video summarization \citep{Mirzasoleiman2018}, and MAP Inference
for Determinantal Point Processes \citep{Gillenwater2012},
among many others.
In recent times, the amount of data generated by many applications has been increasing
exponentially; therefore, linear or sublinear-time algorithms are needed.

If a submodular function $f$ is monotone\footnote{The function $f$ is monotone if for all $A \subseteq B$, $f(A) \le f(B)$.},
greedy approaches for MCC have proven effective 
and nearly optimal, both in terms of query complexity and
approximation factor: subject to a cardinality constraint $k$,
a simple greedy algorithm gives 
a $(1 - 1/e)$ approximation ratio in $O(kn)$ queries \citep{Nemhauser1978},
where $n$ is the size of the instance. Furthermore, this ratio is optimal 
under the value oracle model \citep{Nemhauser1978a}.
\citet{Badanidiyuru2014}  
sped up the greedy algorithm to require 
$O\left( \frac{n}{\epsi} \log \frac{n}{\epsi} \right)$ queries while sacrificing only a small $\epsi > 0$ 
in the approximation ratio, while \citet{Mirzasoleiman2014} developed
a randomized $(1 - 1/e - \epsi)$ approximation in $O(n / \epsi)$ queries.

When $f$ is non-monotone,
the situation is very different; no subquadratic deterministic algorithm
has yet been developed. Although a linear-time, randomized
$(1/e - \epsi)$-approximation has been developed
by \citet{Buchbinder2015a}, which requires $O\left( \frac{n}{\epsi^2} \log \frac{1}{\epsi} \right)$
queries, the performance guarantee of this algorithm holds only in
expectation. A derandomized version of the algorithm with ratio $1/e$ has
been developed by \citet{Buchbinder2018} but has time complexity $O(k^3 n )$.
Therefore, in this work, 
an emphasis is placed
upon the development of nearly linear-time, deterministic approximation algorithms.
%is a gap between the best hardness result and approximation factor. % even when $f$ is unconstrained,
% maximization is NP-hard \citep{Feige2011a}. Furthermore,
%The best ratio obtained for randomized and deterministic algorithms differs;
%In this work, the problem considered is the maximization
%of  %cardinality constraint or
%the intersection of $p$ matroid constraints. % That is,
% given submodular $f: 2^{[n]} \to \reals$ and $k \in [n + 1]$, determine 
% \begin{equation}
%   \argmax_{|X| \le k} f(X).
% \end{equation}
% This fact is significant as
% all randomized approximation ratios discussed in this paper have guarantees that hold only 
% in expectation rather than with high probability;
% derandomization in the context of submodular optimization
% has proved difficult \citep{Buchbinder2018}.
% The fastest algorithms in prior literature for MCC
% are as follows. 
% The fastest deterministic approximation algorithm 
% is the $(1/6 - \epsi)$-approximation\footnote{To get this ratio, the determistic $(1/2-\epsi)$-approximation for unconstrained non-monotone maximization in $O(n / \epsi)$ queries of \citet{Buchbinder2018} is used as a subroutine.}  of \citet{Gupta2010} which requires
% $O( kn )$ queries of the objective function, while the fastest
% % search algorithm of \citet{Lee2010a}, which
% % requires $O( n^4 )$ queries to $f$ to achieve approximation ratio $1/4$. 
%  Since the approximation ratio of $1/e - \epsi$ only holds in expectation,
% in this work 
%Other randomized algorithms achieve better ratios \citet{} by using the continuous 
%extension \ldots but are very inefficient. 

\subsection*{Contributions}
The deterministic approximation algorithm InterlaceGreedy (Alg. \ref{alg:tandem}) is provided
for maximization of a submodular function subject to a cardinality constraint (MCC).
InterlaceGreedy achieves ratio $1/4$ in $O(kn)$ queries to the objective function. %alternating between greedy selection into two disjoint sets. 
A faster version of the algorithm is formulated
in FastInterlaceGreedy (Alg. \ref{alg:fast-tandem}), which achieves ratio $(1/4 - \epsi)$
in $O\left( \frac{n}{\epsi} \log \frac{k}{\epsi} \right)$ queries.
%$(1/4 - \epsi)$-approximation FastInterlaceGreedy (Theorems \ref{thm:ftg}, \ref{thm:ftg-speed}) for MCC 
In Table \ref{table:cc}, the relationship is shown to the fastest deterministic and randomized algorithms for MCC in prior literature.

Both algorithms operate
by interlacing two greedy procedures together in a novel manner; 
that is, the two greedy procedures alternately select elements into disjoint sets 
and are disallowed from selection of the same element. 
This technique is demonstrated first with the interlacing of two
standard greedy procedures in InterlaceGreedy, before interlacing thresholded
greedy procedures developed by \citet{Badanidiyuru2014} for monotone submodular functions
to obtain the algorithm FastInterlaceGreedy.
% The proof adapts the greedy argument for monotone, submodular functions under a matroid
% constraint to non-monotone functions. It requires a re-ordering property of independent
% sets that holds under cardinality constraint but may not hold under more general constraints.

The algorithms are validated in the context of cardinality-constrained
maximum cut and social network monitoring, which are both instances of MCC. In
this evaluation, FastInterlaceGreedy is more than an order of magnitude faster
than the fastest
deterministic algorithm \citep{Gupta2010} and is both faster and obtains better solution
quality than the fastest
randomized algorithm \citep{Buchbinder2015a}. The source code
for all implementations is available at
\url{https://gitlab.com/kuhnle/non-monotone-max-cardinality}.
\begin{table} \caption{Fastest algorithms for cardinality constraint} \centering \label{table:cc}
%\begin{center} \centering contexts
\begin{tabular}{ |c|c|c|c| } 
 \hline
  \textbf{Algorithm} & \textbf{Ratio} & \textbf{Time complexity} & \textbf{Deterministic?} \\ \hline
  FastInterlaceGreedy (Alg. \ref{alg:fast-tandem}) & $1/4 - \epsi$ & $O\left( \frac{n}{\epsi} \log \frac{k}{\epsi} \right)$ & Yes \\ \hline
  %\citet{Mirzasoleiman2018} & $4 p + O\left( \sqrt{p} \right)$ & $O\left(p \sqrt{p} kn \log (k) \right)$ & Yes \\ \hline
  \citet{Gupta2010} & $1/6 - \epsi$ & $O \left(nk + \frac{n}{\epsi}\right)$ & Yes \\ \hline
  \citet{Buchbinder2015a} & $1/e - \epsi$ & $O\left( \frac{n}{\epsi^2} \log \frac{1}{\epsi} \right)$ & No \\ 
 \hline
\end{tabular}
%\end{center}
\end{table}

\paragraph{Organization}
The rest of this paper is organized as follows. 
Related work and preliminaries 
on submodular optimization
are discussed in the rest of this section.
%Preliminaries are discussed in Section \ref{sec:prelim}.
In 
Section \ref{sec:interlace}, InterlaceGreedy and FastInterlaceGreedy are
presented and analyzed. 
Experimental validation is provided in Section \ref{sec:exp}.
%Finally, the paper is concluded in Section \ref{sec:con}.
\subsection*{Related Work} \label{sec:rw}
The literature on submodular optimization comprises many works. 
In this section, a short review of relevant techniques is given for
MCC; that is, maximization of non-monotone, submodular functions
over a ground set of size $n$ with cardinality constraint $k$.
For further information on other types of submodular optimization, 
interested readers are directed to the survey of 
\citet{Buchbinder2018a} and references therein.

%\paragraph{Non-Monotone Maximization}
%\citep{Feige2011a} 
A deterministic local search algorithm was developed by
\citet{Lee2010a}, which achieves ratio $1/4 - \epsi$
in $O(n^4 \log n)$ queries. This algorithm runs two approximate local search
procedures in succession. By contrast, the algorithm FastInterlaceGreedy
employs interlacing of greedy procedures to obtain the same ratio in
$O\left( \frac{n}{\epsi} \log \frac{k}{\epsi} \right)$ queries.
In addition, a randomized local search algorithm was formulated by
\citet{Vondrak2013}, which achieves ratio $\approx 0.309$
in expectation.

\citet{Gupta2010} developed a deterministic, iterated greedy approach,
wherein two greedy procedures are run in succession and an algorithm 
for unconstrained submodular maximization are employed.
This approach requires $O(nk)$ queries and has ratio
$1/(4 + \alpha)$, where $\alpha$ is the inverse ratio of the employed
subroutine for unconstrained, non-monotone submodular maximization;
under the value query model, the smallest possible value for $\alpha$
is 2, as shown by \citet{Feige2011a}, so this ratio is at most $1/6$. 
The iterated greedy approach of \citet{Gupta2010} 
first runs one standard greedy algorithm to completion,
then starts a second standard greedy procedure; this differs from the interlacing 
procedure which runs two greedy procedures concurrently and alternates between the
selection of elements.
The algorithm of \citet{Gupta2010}
is experimentally compared to FastInterlaceGreedy in Section \ref{sec:exp}.
The iterated greedy approach of \citet{Gupta2010} was extended and analyzed under
more general constraints by a series of works: \citet{Mirzasoleiman2016,Feldman2017,Mirzasoleiman2018}.

An elegant randomized greedy algorithm of
\citet{Buchbinder2014}
achieves expected ratio $1/e$ in $O(kn)$ queries for MCC; 
this algorithm
was derandomized by \citet{Buchbinder2018}, but
the derandomized version requires $\oh{k^3n}$ queries. 
The randomized version was sped up in \citet{Buchbinder2015a} to
achieve expected ratio $1/e - \epsi$ and require 
$O\left( \frac{n}{\epsi^2} \log \frac{1}{\epsi} \right)$ queries.
Although this algorithm has better time complexity than FastInterlaceGreedy,
the ratio of $1/e - \epsi$ holds only in expectation, which is much 
weaker than a deterministic approximation ratio. 
The algorithm of \citet{Buchbinder2015a} is experimentally
evaluated in Section \ref{sec:exp}.

Recently, an improvement in the adaptive complexity of MCC was made by
\citet{Balkanskia}. Their algorithm, BLITS, requires $O \left( \log^2 n \right)$
adaptive rounds of queries to the objective, where the queries within each
round are independent of one another and thus can be parallelized easily. 
Previously the best adaptivity was the trivial $O(n)$. However, each round
requires $\Omega( OPT^2 )$ samples to approximate expectations, which for the applications
evaluated in Section \ref{sec:exp} is $\Omega( n^4 )$. 
For this reason, BLITS is evaluated as a heuristic in comparison with the proposed algorithms in Section \ref{sec:exp}. Further improvements in
adaptive complexity have been made by \citet{Fahrbach2018a}
and \citet{Ene2019}. 

Streaming algorithms for MCC make only
one or a few passes through the ground set.
Streaming algorithms for MCC include
those of \citet{Chekuri2015,Feldman2018,Mirzasoleiman2018}.
A streaming algorithm with low adaptive complexity 
has recently been developed by \citet{Kazemi2019}.
In the following, the algorithms are allowed to make
an arbitrary number of passes through the data.

Currently, the best approximation ratio of any algorithm for MCC is
$0.385$ of \citet{Buchbinder2016}. Their algorithm also
works under a more general constraint than cardinality constraint;
namely, a matroid constraint.
This algorithm is the latest in a series of works (e.g. \citep{Naor2011,Ene2016a}) using
the multilinear extension of a submodular function,
which is expensive to evaluate. 

%\citep{Vondrak2008a} $0.309$
%\citep{Gharan2010} $0.325$
%\citep{Naor2011} $1/e - o(1)$
%cts. greedy \citep{Calinescu2011}

% \paragraph{On Derandomization}
% Standard derandomization techniques do not work in the
% context of submodular maximization. 
% % Even for monotone submodular maximization,
% % there is disparity subject to a single matroid constraint,
% % for which the best deterministic ratio is $1/2$ \citep{Fisher1978} and
% % the best randomized ratio is $1 - 1/e$ \citep{Calinescu2007}.
% Recently,
% \citet{Buchbinder2018} came up with a procedure
% to derandomize certain algorithms 
% at the cost of increased query complexity. 
% For example, 
% However, this approach does not work for algorithms
% that use the continuous extension \citep{}; currently,
% the best approximation ratio for a randomized algorithm is
% $X$ \citep{}, while the best ratio for a deterministic
% algorithm is $1/e$. 

\subsection*{Preliminaries} \label{sec:prelim}
Given $n \in \mathbb N$, the notation $[n]$ is used
for the set $\{0,1,\ldots,n-1\}$. In this work,
functions $f$ with domain all subsets of a finite set are considered;
hence, without loss of generality, the domain
of the function $f$ is taken to be $2^{[n]}$,
which is all subsets of $[n]$.
An equivalent characterization of submodularity is that
for each $A, B \subseteq [n]$, $f( A \cup B ) + f( A \cap B ) \le f(A) + f(B)$.
For brevity, the notation $f_x( A )$ is used to denote the marginal gain
$f(A \cup \{ x \}) - f(A)$ of adding element $x$ to set $A$.
%The notation $f|_A$ is used to indicate the function that is $f$ restricted to $A$.

In the following, the problem studied is to maximize a submodular function
under a cardinality constraint (MCC), which is formally defined as follows.
Let $f:2^{n} \to \reals$ be submodular; let $k \in [n]$. Then the problem is to determine
$$ \argmax_{A \subseteq [n] : |A| \le k } f(A). $$
%$A \subseteq [n]$ such that $|A| \le k$ and for all $B$ such that 
%$|B| \le k$, $f(B) \le f(A)$. 
An instance of MCC is the pair $(f,k)$;
however, rather than an explicit description of $f$, the 
function $f$ is accessed by
a value oracle; the value oracle 
may be queried on any set $A \subseteq [n]$ to yield
$f(A)$. The efficiency or runtime of an algorithm is measured by the number of queries made to the oracle for $f$.

Finally, without loss of generality, instances of MCC considered in the following
satisfy $n \ge 4k$. If this condition does not hold, the function may be extended to 
$[m]$ by adding dummy elements to the domain which do not change the function value.
That is, the function $g:2^{m} \to \reals$ is defined as $g( A ) = f( A \cap [n] )$;
it may be easily checked that $g$ remains submodular, and any possible solution to the
MCC instance $(g,k)$ maps\footnote{The mapping is to discard all elements greater than $n$.} 
to a solution of $(f,k)$ of the same value. Hence, 
the ratio of any solution to $(g,k)$
to the optimal is the same as the ratio of the mapped solution to the optimal on $(f,k)$.
\section{Approximation Algorithms} \label{sec:interlace}
In this section, the approximation algorithms based upon interlacing
greedy procedures are presented. In Section \ref{sec:slowInterlace}, the technique
is demonstrated with standard greedy procedures in algorithm InterlaceGreedy.
In Section \ref{sec:fig}, the nearly
linear-time algorithm FastInterlaceGreedy is introduced.
\subsection{The InterlaceGreedy Algorithm} \label{sec:slowInterlace}
In this section, the InterlaceGreedy algorithm (\ig, Alg. \ref{alg:tandem}) is introduced. 
\ig takes as input an instance of MCC and
outputs a set $C$.

\begin{algorithm}
   \caption{\ig$(f,k)$: The InterlaceGreedy Algorithm}
   \label{alg:tandem}
   \begin{algorithmic}[1]
     \STATE {\bfseries Input:} $f : 2^{[n]} \to \reals$, 
     $k \in [n]$
     \STATE {\bfseries Output:} $C \subseteq [n]$, such that $|C| \le k$.
     \STATE $A_0 \gets B_0 \gets \emptyset$
     
     \FOR{$i \gets 0$ to $k - 1$}
     \STATE $a_i \gets \argmax_{x \in [n] \setminus (A_i \cup B_i)} f_x( A_{i} )$ 
     \STATE $A_{i+1} \gets A_i + a_i$
     \STATE $b_i \gets \argmax_{x \in [n] \setminus (A_{i+1} \cup B_i)} f_x( B_{i} )$ 
     \STATE $B_{i+1} \gets B_i + b_i$
     \ENDFOR
     \STATE $D_1 \gets E_1 \gets \{ a_0 \}$
     \FOR{$i \gets 1$ to $k - 1$}
     \STATE $d_i \gets \argmax_{x \in [n] \setminus (D_i \cup E_i)} f_x( D_{i} )$ 
     \STATE $D_{i+1} \gets D_i + d_i$
     \STATE $e_i \gets \argmax_{x \in [n] \setminus (D_{i+1} \cup E_i)} f_x( E_{i} )$ 
     \STATE $E_{i+1} \gets E_i + e_i$
     
     \ENDFOR
   \STATE \textbf{return} $C \gets \argmax \{f(A_i), f(B_i), f(D_i), f(E_i) : i \in [k + 1] \}$
\end{algorithmic}
\end{algorithm}
\ig operates by interlacing two standard greedy procedures. 
This interlacing is accomplished by maintaining two disjoint sets $A$ and $B$, which
are initially empty. For
$k$ iterations, the element $a \not \in B$ with the highest marginal gain 
with respect to $A$ is added to $A$, followed by an analogous greedy selection for $B$;
that is, the element $b \not \in A$ with the highest marginal gain 
with respect to $B$ is added to $B$.
After the first set of interlaced greedy procedures complete, a modified 
version is repeated with sets $D,E$, which are initialized to the maximum-value singleton $\{ a_0 \}$. 
Finally, the algorithm returns the set with the maximum $f$-value of any query the
algorithm has made to $f$.

If $f$ is submodular, InterlaceGreedy has an approximation ratio of $1/4$ and query complexity
$O( kn )$; the deterministic algorithm of \citet{Gupta2010} 
has the same time complexity to achieve ratio $1/6$. The full proof
of Theorem \ref{thm:tg} 
% is similar to the proof of Theorem \ref{thm:ftg} for the faster version of InterlaceGreedy;
% it
is provided in Appendix \ref{apx:proof-thm-tg}. 
\begin{theorem} \label{thm:tg}
  Let $f: 2^{[n]} \to \reals$ be submodular, let $k \in [ n ]$,
  let $O = \argmax_{|S| \le k} f(S)$, and let $C = $ \ig$(f,k)$.
  Then 
  \[ f(C) \ge  f(O) / 4, \]
  and \ig makes $O( kn )$ queries to $f$. 
\end{theorem}
\begin{proof}[Proof sketch]
  The argument of 
  \citet{Fisher1978} shows that
  the greedy algorithm is a $(1/2)$-approximation 
  for monotone submodular maximization with respect to a matroid
  constraint. This argument also applies 
  to non-monotone, submodular functions, but it shows only that
  $f(S) \ge \frac{1}{2} f(O \cup S)$, where $S$ is returned by
  the greedy algorithm. Since $f$ is non-monotone, it is possible
  for $f( O \cup S ) < f(S)$. 
  The main idea of the InterlaceGreedy algorithm 
  is to exploit the fact that 
  if $S$ and $T$ are disjoint, 
  \begin{equation} \label{eq:sm_disjoint}
  f( O \cup S ) + f( O \cup T) \ge f( O ) + f( O \cup S \cup T ) \ge f(O),
\end{equation}
 which
  is a consequence of the submodularity of $f$. Therefore,
  by interlacing two greedy procedures, two disjoint
  sets $A$,$B$ are obtained, which can be shown to almost satisfy 
  $f(A) \ge \frac{1}{2}f( O \cup A )$ and
  $f(B) \ge \frac{1}{2}f( O \cup B )$,
  after which the result follows from (\ref{eq:sm_disjoint}).
  There is a technicality wherein the element $ a_0 $ must
  be handled separately, which requires the second round of
  interlacing to address.
  %  Because of the way the re-ordering works, it is only possible to show that
  % $f(B) \ge f( (O \setminus \{a_0 \}) \cup B ) / 2$, instead of the desired $f(B) \ge f(O \cup B) /2$.
  % Hence, a second greedy interlacing is required, starting both sets from $\{ a_0 \}$, to produce
  % $D,E$ such that $f(D) \ge f( O \cup D ) / 2$ and $f(E) \ge f(O \cup E) / 2$, with 
  % $f( O \cup D ) + f( O \cup E ) \ge f( O \cup \{ a_0 \} )$ by submodularity. Finally,
  % the argument concludes by noticing that either $a_0 \in O$ or $a_0 \not \in O$.
\end{proof}
\subsection{The FastInterlaceGreedy Algorithm} \label{sec:fig}
In this section, a faster interlaced greedy algorithm
(FastInterlaceGreedy (\fig), Alg. \ref{alg:fast-tandem}) 
is formulated, which requires $O(n \log k)$ queries. 
As input, an instance $(f,k)$ of MCC is taken, as well as a parameter $\delta > 0$.

\begin{algorithm}
  \caption{\fig$(f,k, \delta)$: The FastInterlaceGreedy Algorithm}
   \label{alg:fast-tandem}
   \begin{algorithmic}[1]
     \STATE {\bfseries Input:} $f : 2^{[n]} \to \reals$, 
     $k \in [ n ]$
     \STATE {\bfseries Output:} $C \subseteq [n]$, such that $|C| \le k$.
     \STATE $A_0 \gets B_0 \gets \emptyset$
     \STATE $M \gets \tau_A \gets \tau_B \gets \max_{x \in [n]} f(x)$
     \STATE $i \gets -1$, $a_{-1} \gets 0$, $b_{-1} \gets 0$
     \WHILE{$\tau_A \ge \stopGain$ or $\tau_B \ge \stopGain$}
     \STATE $(a_{i+1}, \tau_A) \gets \add (A,B,a_i,\tau_A)$
     \STATE $(b_{i+1}, \tau_B) \gets \add (B,A,b_i,\tau_B)$
     \STATE $i \gets i + 1$
     \ENDWHILE
     \STATE $D_1 \gets E_1 \gets \{ a_0 \}$, $\tau_D \gets \tau_E \gets M$
     \STATE $i \gets 0$, $d_{0} \gets 0$, $e_{0} \gets 0$
     \WHILE{$\tau_D \ge \stopGain$ or $\tau_E \ge \stopGain$}
     \STATE $(d_{i+1}, \tau_D) \gets \add (D,E,d_i,\tau_D)$
     \STATE $(e_{i+1}, \tau_E) \gets \add (E,D,e_i,\tau_E)$
     \STATE $i \gets i + 1$
     \ENDWHILE
   \STATE \textbf{return} $C \gets \argmax \{f(A), f(B), f(D), f(E) \}$
\end{algorithmic}
\end{algorithm}
\begin{algorithm}
\caption{\add$(S,T,j,\tau)$: The \add subroutine}
   \label{alg:add}
   \begin{algorithmic}[1]
     \STATE {\bfseries Input:} Two sets $S,T \subseteq [n]$, element $j \in [n]$, $\tau \in \reals$
     \STATE {\bfseries Output:} $(i, \tau)$, such that $i \in [n]$, $\tau \in \reals$
     \IF{$|S| = k$}
     \STATE \textbf{return} $( 0, (1 - \delta)\tau )$
     \ENDIF
     \WHILE{$\tau \ge \stopGain$}
     \FOR{$(x \gets j; x < n; x \gets x + 1)$}
     \IF{$x \not \in T$}
     \IF{$f_x(S) \ge \tau$}
     \STATE $S \gets S \cup \{ x \}$
     \STATE \textbf{return} $( x, \tau )$
     \ENDIF
     \ENDIF
     \ENDFOR
     \STATE $\tau \gets (1 - \delta) \tau$
     \STATE $j \gets 0$
     \ENDWHILE
   \STATE \textbf{return} $(0 , \tau )$
\end{algorithmic}
\end{algorithm}
The algorithm \fig works as follows. As in InterlaceGreedy, there is
a repeated interlacing of two
greedy procedures. However, to ensure a faster query complexity,
these greedy procedures are thresholded: a separate threshold
$\tau$ is maintained for each of the greedy procedures. The interlacing
is accomplished by alternating calls to the \add subroutine (Alg. \ref{alg:add}),
which adds a single element and is described below. When all of the thresholds fall below the value $\delta M / k$,
the maximum of the greedy solutions is returned; here, $\delta > 0$ is the input parameter,
$M$ is the maximum value of a singleton, and $k \le n$ is the cardinality constraint.

The \add subroutine is responsible for adding a single element above the input threshold
and decreasing the threshold. It takes as input four parameters: two sets $S,T$, element $j$,
and threshold $\tau$; furthermore, \add is given access to the oracle $f$, the budget $k$,
and the parameter $\delta$ of \fig. 
As an overview, \add adds the first\footnote{The first element $x > j$ in the natural ordering on $[n] = \{0, \ldots, n - 1\}$.} element $x \ge j$, such that $x \not \in T$ and
such that the marginal gain $f_x(S)$ is at least $\tau$. If no such element $x \ge j$ exists,
the threshold is decreased by a factor of $(1 - \delta)$ and the process is repeated (with $j$
set to $0$). When such an element $x$ is found, the element $x$ is added to $S$, and 
the new threshold value and position $x$ are returned. Finally, \add ensures that the
size of $S$ does not exceed $k$.

Next, the approximation ratio of \fig is proven.
\begin{theorem} \label{thm:ftg}
  Let $f: 2^{[n]} \to \reals$ be submodular, let $k \in [n]$,
  and let $\epsi > 0$. 
  Let $O = \argmax_{|S| \le k} f(S)$.
  Choose $\delta$ such that $(1 - 6 \delta)/4 > 1/4 - \epsi$,
  and let $C = $ \fig$(f,k,\delta )$.
  Then 
  \[ f(C) \ge ( 1 - 6\delta)f(O) / 4 \ge \left(1/4 - \epsi \right) f(O). \]
\end{theorem}
\begin{proof}\let\qed\relax
Let $A,B,C,D,E,M$ have their values at termination of $\fig (f,k,\delta )$. Let $A = \{ a_0, \ldots, a_{|A| - 1} \}$ 
be ordered by addition of elements by \fig into $A$. The proof
requires the following four inequalities:
\begin{align}
  f(O \cup A) &\le (2 + 2\delta )f(A) + \delta M, \label{ineq:fast-A} \\
  f( (O \setminus \{ a_0 \}) \cup B ) &\le (2 + 2\delta )f(B) + \delta M, \label{ineq:fast-B} \\
  f( O \cup D ) &\le (2 + 2\delta )f(D) + \delta M, \label{ineq:fast-D}\\
  f( O \cup E ) &\le (2 + 2\delta )f(E) + \delta M. \label{ineq:fast-E}
\end{align}
Once these inequalities have been established, Inequalities \ref{ineq:fast-A},
\ref{ineq:fast-B}, submodularity of $f$, and $A \cap B = \emptyset$ imply 
\begin{equation}
  f( O \setminus \{a_0\} ) \le 2(1 + \delta)(f(A) + f(B)) + 2 \delta M.
\end{equation}
Similarly, from Inequalities \ref{ineq:fast-D}, \ref{ineq:fast-E},
submodularity of $f$, and $D \cap E = \{ a_0 \}$, it holds that
\begin{equation}
  f( O \cup \{ a_0 \} ) \le 2(1 + \delta)(f(D) + f(E)) + 2 \delta M.
\end{equation}
Hence, from the fact that either $a_0 \in O$ or $a_0 \not \in O$
and the definition of $C$, it holds that
\begin{equation*}
  f(O) \le 4(1+\delta )f(C) + 2 \delta M.
\end{equation*}
Since $f(C) \le f(O)$ and $M \le f(O)$, the theorem is proved.

The proofs of Inequalities \ref{ineq:fast-A}--\ref{ineq:fast-E} 
are similar. The proof of Inequality \ref{ineq:fast-B} is given here,
while the proofs of the others are provided in Appendix \ref{apx:proof-thm-ftg}.
\begin{proof}[Proof of Inequality \ref{ineq:fast-B}] 
Let $A = \{a_0, \ldots, a_{|A| - 1} \}$ be ordered as specified
by \fig. Likewise, let
$B = \{b_0, \ldots, b_{|B| - 1} \}$ be ordered as specified
by \fig.
\begin{lemma} \label{lemma:fast-order-B}
  $O \setminus ( B \cup \{a_0\} ) = \{o_0,\ldots,o_{l-1} \}$ can be ordered such that
  \begin{equation} \label{ineq:fast-marg-B}
    f_{o_i}(B_i) \le (1 + 2 \delta) f_{b_i}(B_i),
  \end{equation}
  for any $i \in [|B|]$.
\end{lemma}
\begin{proof}
  For each $i \in [|B|]$, define $\tau_{B_i}$ to be the value of $\tau$ when $b_i$
  was added into $B$ by the \add subroutine.
  Order $o \in (O \setminus (B \cup \{ a_0 \} )) \cap A = \{o_0, \ldots, o_{\ell - 1} \} $ by
  the order in which these elements were added into $A$. Order the remaining
  elements of $O \setminus (B \cup \{ a_0 \})$ arbitrarily. Then, when 
  $b_i$ w;as chosen by \add, it holds that $o_i \not \in A_{i + 1}$, since
  $A_1 = \{ a_0 \}$ and $a_0 \not \in O \setminus (B \cup \{ a_0 \})$.
  Also, it holds that $o_i \not \in B_i$ since $B_i \subseteq B$;
  hence $o_i$ was not added into some (possibly non-proper) subset $B'_i$ of $B_i$ 
  at the previous threshold value $\frac{\tau_{B_i}}{ (1 - \delta)}$.
  By submodularity,
  $f_{o_i}(B_i) \le f_{o_i}(B'_i) < \frac{\tau_{B_i}}{(1 - \delta)}$.
  Since $f_{b_i}(B_i) \ge \tau_{B_i}$ and $\delta < 1/2$, 
  inequality (\ref{ineq:fast-marg-B}) follows.
\end{proof}
  Order $\hat{O} = O\setminus (B\cup \{a_0\}) = \{o_0, \ldots, o_{l-1} \}$ as 
  defined in the proof of Lemma \ref{lemma:fast-order-B},
  and let $\hat{O}_i = \{o_0, \ldots, o_{i-1} \}$, if $i \ge 1$, and
  let $\hat{O}_0 = \emptyset$.
  Then 
  \begin{align*}
    f( \hat{O}  \cup B ) - f( B ) &= \sum_{i=0}^{l-1} f_{o_i} ( \hat{O}_i \cup B )\\
    &= \sum_{i=0}^{|B| - 1} f_{o_i} ( \hat{O}_i \cup B ) + \sum_{i=|B|}^{l - 1} f_{o_i} ( \hat{O}_i \cup B ) \\
    &\le \sum_{i=0}^{|B| - 1} f_{o_i} ( B_i ) + \sum_{i=|B|}^{l - 1} f_{o_i} (B  ) \\
    &\le \sum_{i=0}^{|B| - 1} (1 + 2\delta ) f_{b_i} ( B_i ) + \sum_{i=|B|}^{l - 1} f_{o_i} ( B ) \\
    &\le ( 1 + 2 \delta ) f(B) + \delta M,
  \end{align*}
  where any empty sum is defined to be 0; the first inequality follows by submodularity, the second follows from Lemma \ref{lemma:fast-order-B}, and the third follows from the definition of $B$, 
  and the facts that, for any $i$ such that $|B| \le i < l$, $\max_{x \in [n] \setminus A_{|B|+1}} f_x(B) < \stopGain $, $l - |B| \le k$, and $o_i \not \in A_{|B| + 1}$.
\end{proof}\end{proof}
\begin{theorem} \label{thm:ftg-speed}
  Let $f: 2^{[n]} \to \reals$ be submodular, let $k \in [n ]$,
  and let $\delta > 0$.
  Then the number of queries to $f$ by $\fig(f,k,\delta)$ is
  at most $O\left( \frac{n}{\delta} \log \frac{k}{\delta} \right)$.
\end{theorem}
\begin{proof}
  Recall $[n] = \{0, 1, \ldots, n - 1\}$. Let $S \in \{ A,B,D,E \}$, and 
  $S = \{s_0, \ldots, s_{|S| - 1} \}$ in the order in which elements were
  added to $S$.
  When \add is called by \fig to add an element $s_i \in [n]$ to $S$,
  if the value of $\tau$ is the same as the value when $s_{i-1}$ was added to $S$, then 
  $s_i > s_{i-1}$. Finally, once \add queries the marginal gain of adding $(n - 1)$,
  the threshold is revised downward by a factor of $(1 - \delta)$. 
  
  Therefore, there are at most $O(n)$ queries of $f$ at each distinct value
  of $\tau_A$, $\tau_B$, $\tau_D$, $\tau_E$.  Since at
  most $O(\frac{1}{\delta} \log \frac{k}{\delta})$ values are assumed by each of these
  thresholds, the theorem follows. 
\end{proof}
\section{Tight Examples}
In this section, examples are provided 
showing that InterlaceGreedy or FastInterlaceGreedy may achieve performance ratio at most
$1/4 + \epsi$ on specific instances, for each $\epsi > 0$. These examples show that the
analysis in the preceding sections is tight. 

Let $\epsi > 0$ and choose $k$ such that $1/k < \epsi$. Let $O$ and $D$
be disjoint sets each of $k$ distinct elements; and let $U = O \dot{\cup} \{a,b\} \dot{\cup} D$.
A submodular function $f$ will be defined on subsets of $U$ as follows.

Let $C \subseteq U$.
\begin{itemize}
  \item If both $a \in C$ and $b \in C$, then $f(C) = 0$.
  \item If $a \in C$ xor $b \in C$, then $f(C) = \frac{|C \cap O|}{2k} + \frac{1}{k}$.
  \item If $a \not \in C$ and $b \not \in C$, then $f(C) = \frac{|C \cap O|}{k}$.
\end{itemize}

The following proposition is proved in Appendix \ref{apx:tight}.
\begin{proposition}\label{prop:sm}
  The function $f$ is submodular.
\end{proposition}

Next, observe that for any $o \in O$, 
$f_{a}( \emptyset ) = f_b( \emptyset) = f_o( \emptyset ) = 1/k$. Hence InterlaceGreedy or FastInterlaceGreedy may choose $a_0 = a$ and $b_0 = b$; after this choice,
the only way to increase $f$ is by choosing elements of $O$. Hence $a_i, b_i$ will be chosen in $O$
until elements of $O$ are exhausted, which results in $k/2$ elements of $O$ added
to each of $A$ and $B$.  Thereafter, elements of $D$ will be chosen, which do not
affect the function value. This yields 
$$f(A) = f(B) \le 1/k + 1/4.$$
Next, $D_1 = E_1 = \{ a \}$, and a similar situation arises, in which 
$k/2$ elements of $O$ are added to $D,E$, yielding $f(D) = f(E) = f(A)$. 
Hence InterlaceGreedy or FastInterlaceGreedy may return $A$,
while $f( O ) = 1$. So 
$\frac{f(A)}{f(O)} \le 1/k + 1/4 \le 1/4 + \epsi$. 
\section{Experimental Evaluation} \label{sec:exp}
In this section, performance of FastInterlaceGreedy (\fig) 
is compared with that of state-of-the-art algorithms on
two applications of
submodular maximization: cardinality-constrained maximum cut and network monitoring. 
\subsection{Setup}
\paragraph{Algorithms}
The following algorithms are compared.
Source code for the evaluated implementations of all algorithms is available at \url{https://gitlab.com/kuhnle/non-monotone-max-cardinality}.
\begin{itemize}
\item \textbf{ FastInterlaceGreedy (Alg. \ref{alg:fast-tandem})}: \fig is implemented as specified
    in the pseudocode, with the following addition: a stealing procedure
    is employed at the end, which uses submodularity to quickly steal\footnote{Details of the stealing procedure are given in Appendix \ref{apx:steal}.}
    elements from $A,B,D,E$ into $C$ in $O(k)$ queries. This does
    not impact the performance guarantee, as the value of $C$ can only
    increase. The parameter $\delta$ is set to $0.1$, yielding
    approximation ratio of $0.1$.
    \item \textbf{\citet{Gupta2010}}: The algorithm of \citet{Gupta2010} for cardinality
    constraint; as the subroutine for the unconstrained maximization subproblems,
    the deterministic, linear-time $1/3$-approximation algorithm of
    \citet{Naor2012} is employed. This yields an overall approximation ratio
    of $1/7$ for the implementation used herein. This algorithm is the fastest
    determistic approximation algorithm in prior literature.
    \item \textbf{FastRandomGreedy (FRG}): The $O \left( \frac{n}{\epsi^2} \ln \frac{1}{\epsi} \right)$ 
    randomized algorithm of \citet{Buchbinder2015a} (Alg. 4 of that paper), with
    expected ratio $1/e - \epsi$; the parameter $\epsi$ was set to 0.3, yielding
    expected ratio of $\approx 0.07$ as evaluated herein.
    This algorithm is the fastest
    randomized approximation algorithm in prior literature.
    \item \textbf{BLITS}: The $O \left( \log^2 n \right)$-adaptive algorithm
    recently introduced in \citet{Balkanskia}; the algorithm is employed
    as a heuristic without performance ratio, with the same parameter choices as in \citet{Balkanskia}. 
    In particular, $\epsi = 0.3$
    and 30 samples are used to approximate the expections. Also, a bound on OPT
    is guessed in logarithmically many iterations as described in \citet{Balkanskia}
    and references therein.
\end{itemize}
Results for randomized algorithms are the mean of 10 trials, and the standard deviation
is represented in plots by a shaded region.
\paragraph{Applications}
Many applications with non-monotone, submodular objective functions exist. In this
section, two applications are chosen to demonstrate the performance of the evaluated
algorithms.
\begin{itemize}
  \item Cardinality-Constrained Maximum Cut: The archetype of a submodular, non-monotone
    function 
    is the maximum cut objective: given graph $G = (V,E)$, $S \subseteq V$, $f(S)$ is
    defined to be the number of edges crossing from $S$ to $V \setminus S$. The cardinality constrained version of this problem is considered in the evaluation.
  \item Social Network Monitoring: Given an online social 
    network, suppose it is desired to choose $k$ users to monitor,
    such that the maximum amount of content is propagated 
    through these users. 
    Suppose the amount of content propagated between two users $u,v$ is encoded as
    weight $w(u,v)$. Then
    $ f(S) = \sum_{u \in S, v \not \in S} w(u,v).$
\end{itemize}
\subsection{Results} In this section, results are presented for the algorithms on the two 
applications. In overview: in terms of objective value, \fig and \citet{Gupta2010} were
about the same and outperformed BLITS and FRG. Meanwhile, \fig was the fastest algorithm
by the metric of queries to the objective and was faster than \citet{Gupta2010} by at least
an order of magnitude.
\paragraph{Cardinality Constrained MaxCut}
\begin{figure}[h]
  \subfigure[ER, Cut Value]{ \label{er:obj}
    \includegraphics[width=0.31\textwidth,height=0.15\textheight]{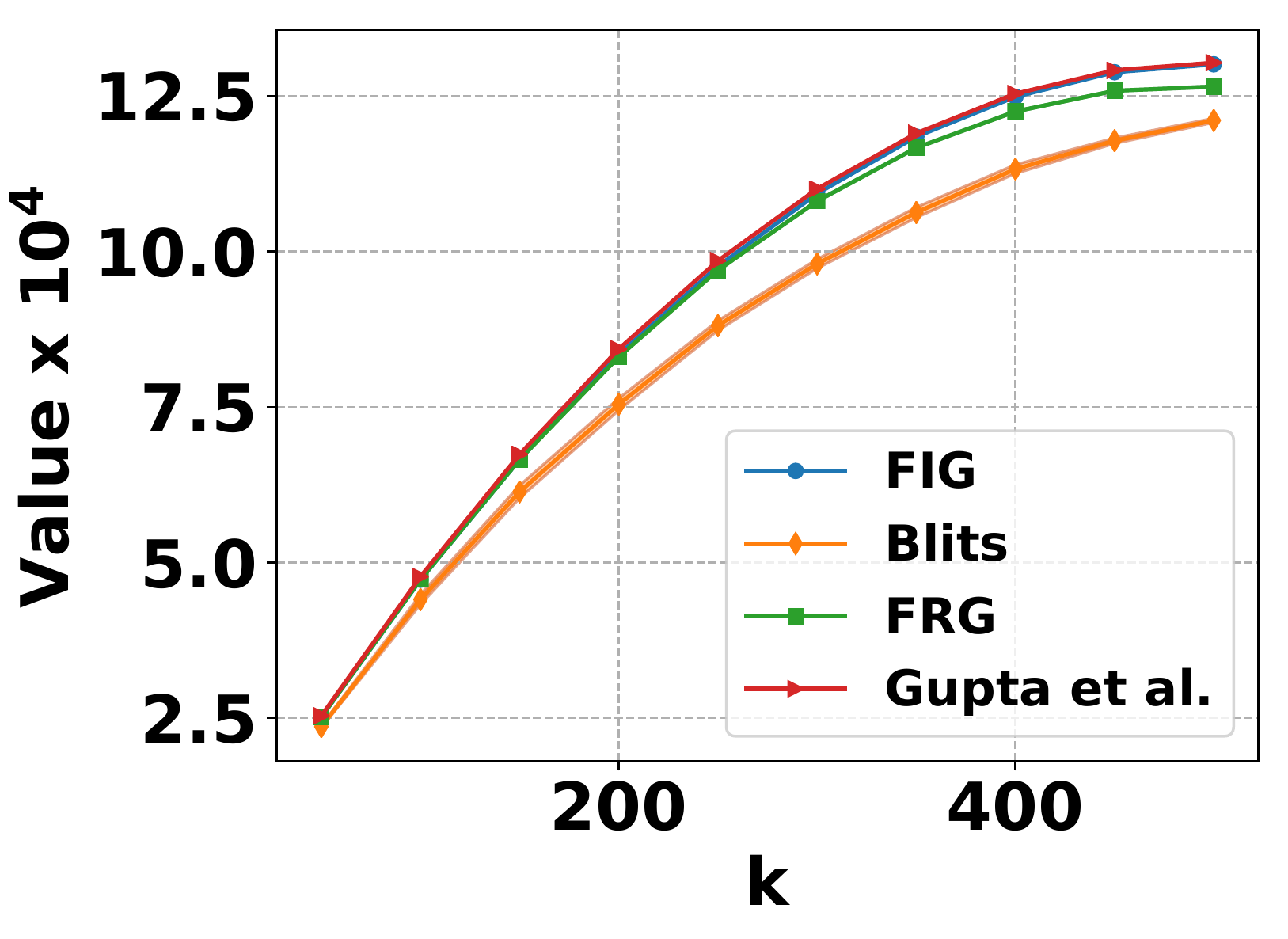}
  }
  \subfigure[ER, Function Queries]{ \label{er:que}
    \includegraphics[width=0.31\textwidth,height=0.15\textheight]{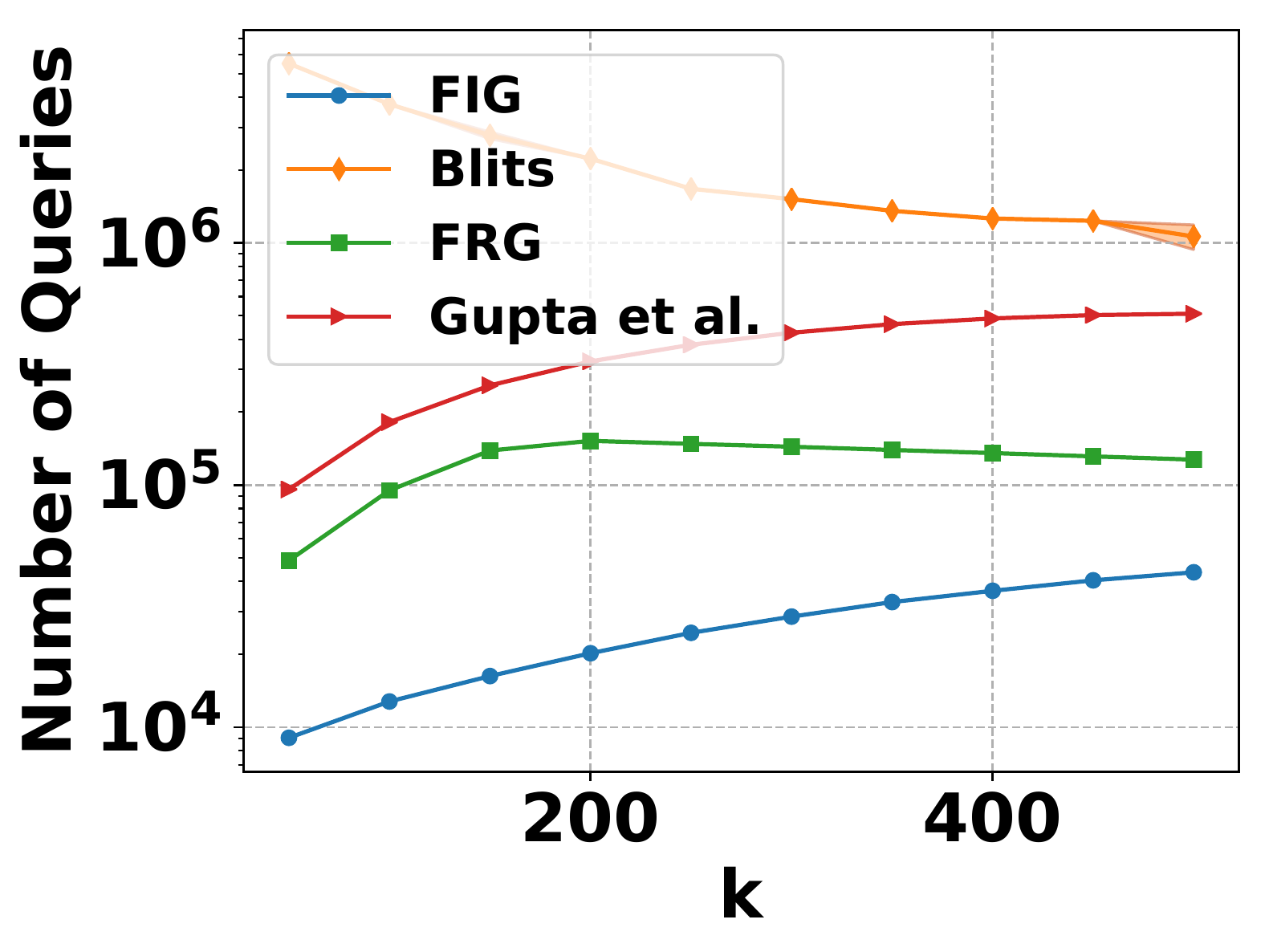}
  }
  \subfigure[BA, Cut Value]{ \label{ba:obj}
    \includegraphics[width=0.31\textwidth,height=0.15\textheight]{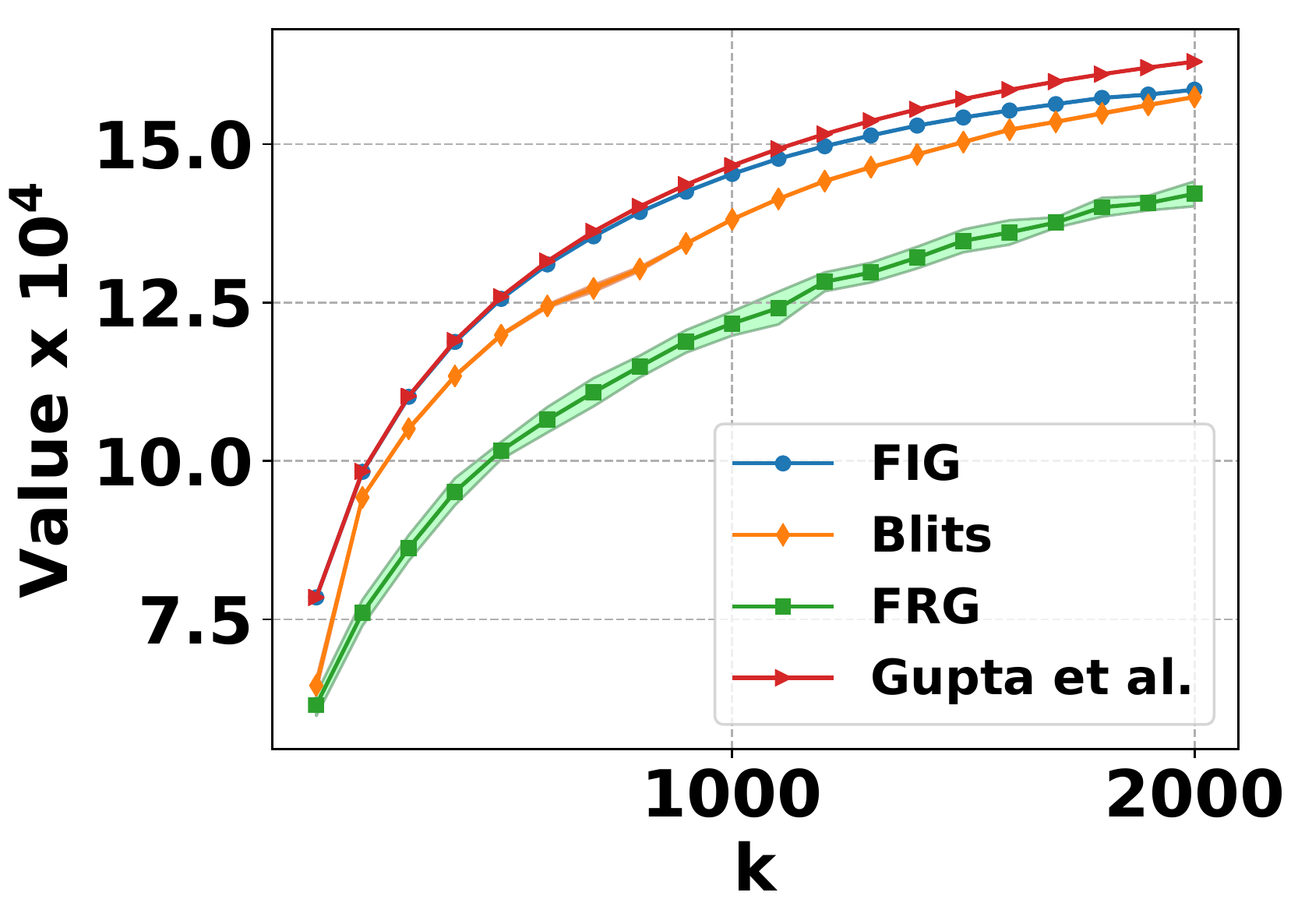}
  }

  \subfigure[BA, Function Queries]{ \label{ba:que}
    \includegraphics[width=0.31\textwidth,height=0.15\textheight]{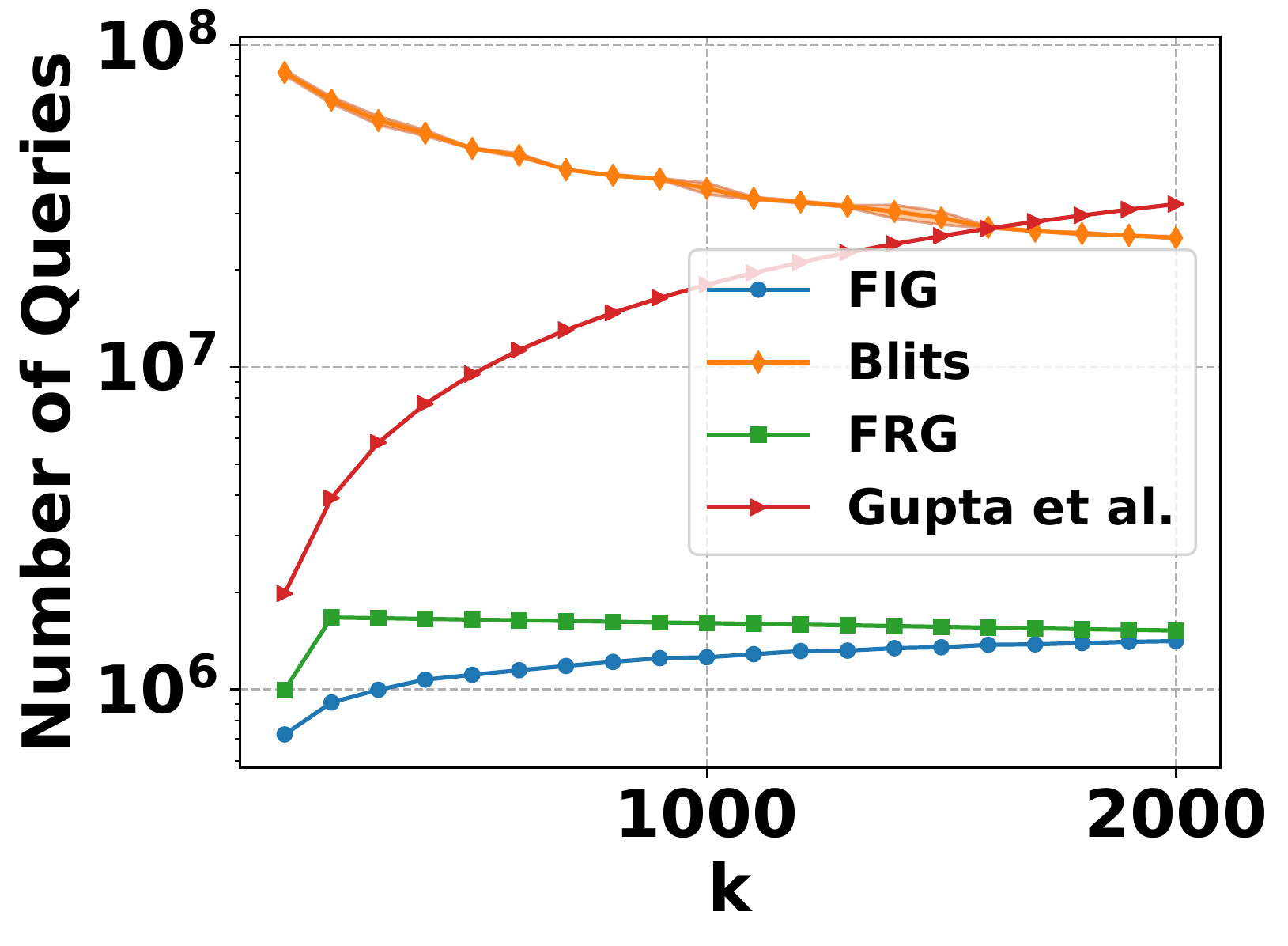}
  }
  \subfigure[Total content monitored versus budget $k$]{
    \includegraphics[width=0.31\textwidth,height=0.15\textheight]{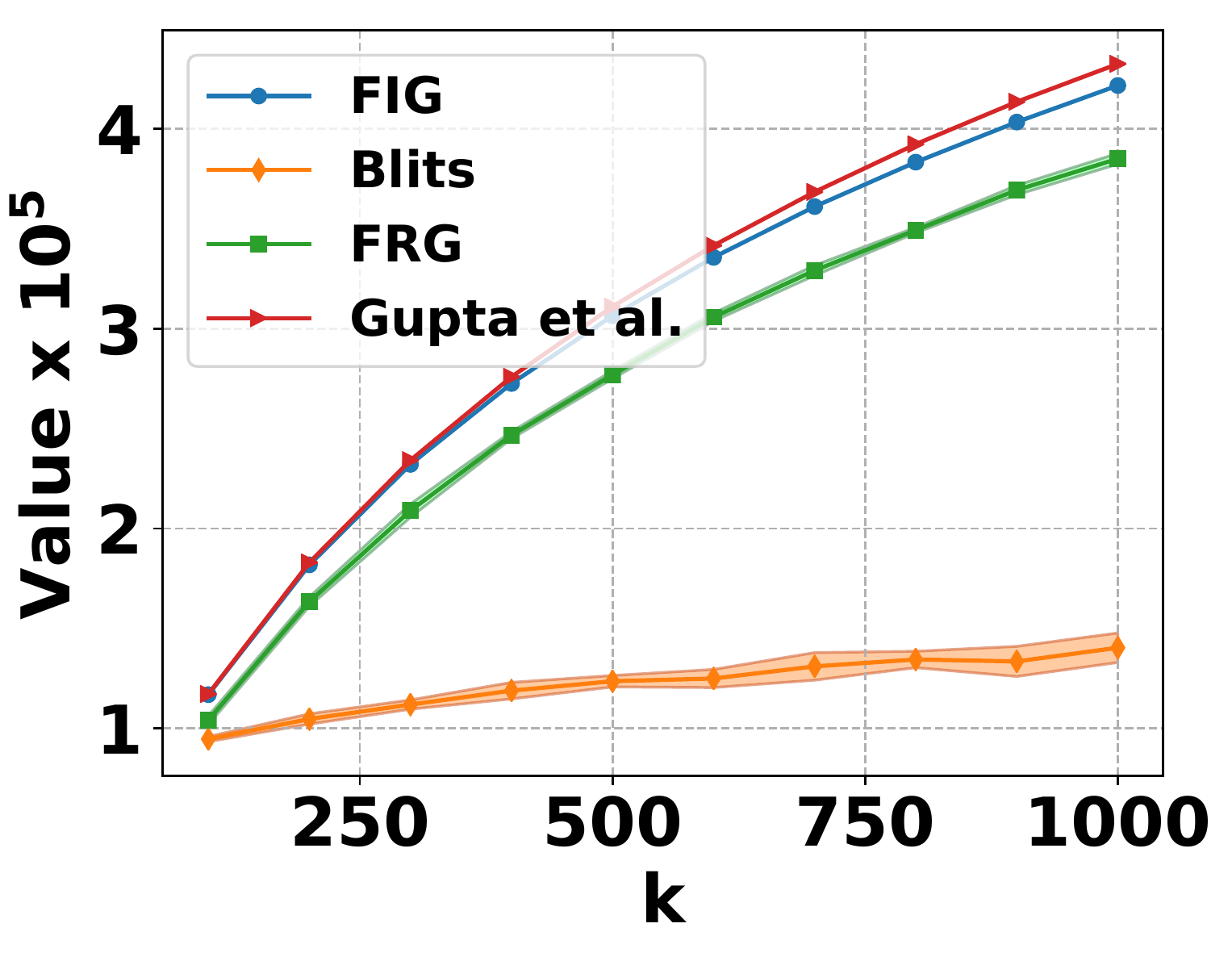}
  }
  \subfigure[Number of Queries versus budget $k$]{
    \includegraphics[width=0.31\textwidth,height=0.15\textheight]{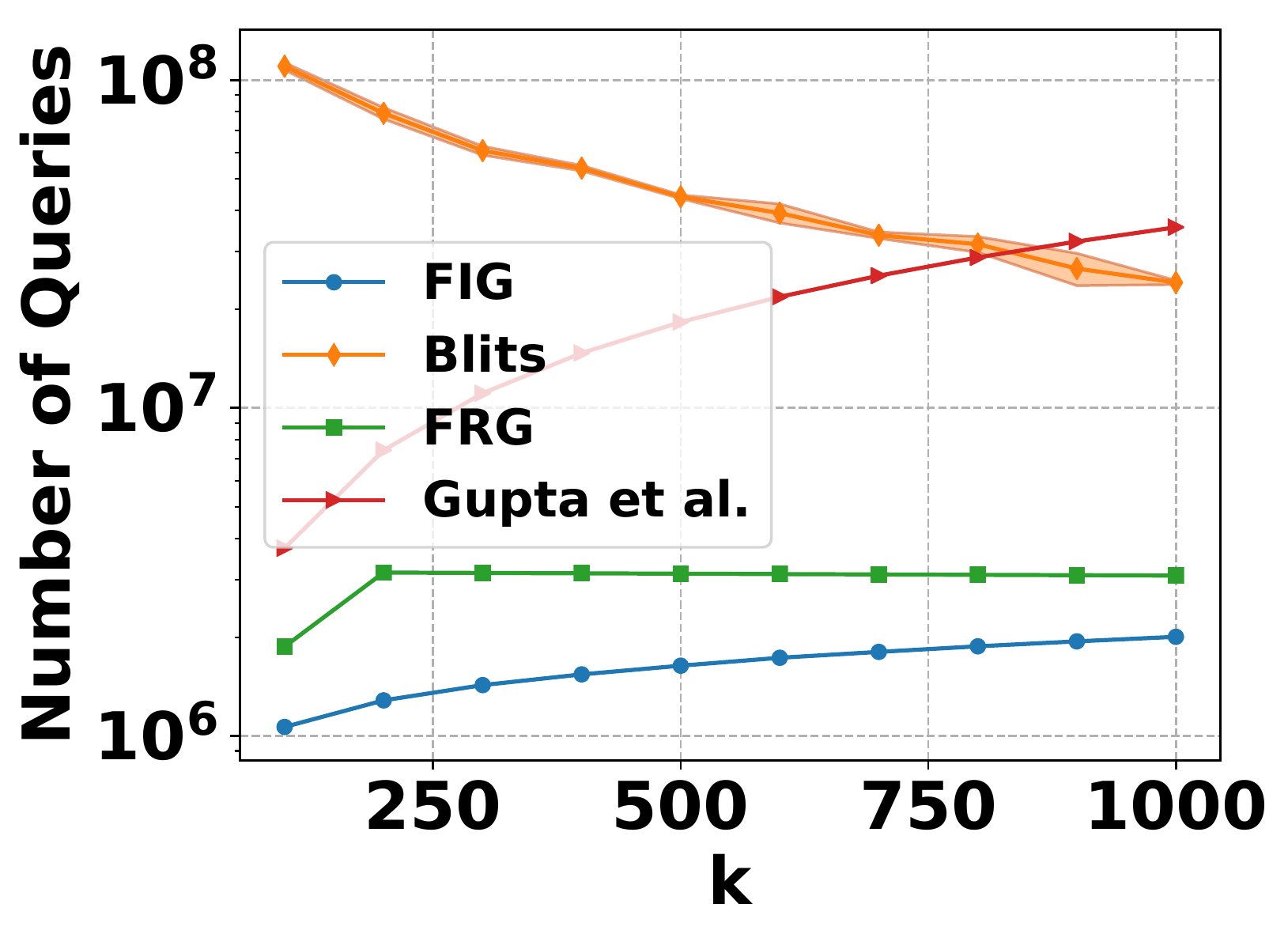}
  }
  \caption{\textbf{(a)--(d)}: Objective value and runtime for cardinality-constrained maxcut on random graphs. \textbf{(e)--(f)}:
Objective value and runtime for cardinality-constrained maxcut on ca-AstroPh with simulated amounts of content between users. In all plots, the $x$-axis shows the budget $k$. }
\end{figure}
%\begin{figure}[h]
  % \subfigure[]{
  %   \includegraphics[width=0.22\textwidth,height=0.1\textheight]{./fig/randomGraphs/obj-er}
  % }
  % \subfigure[]{
  %   \includegraphics[width=0.22\textwidth,height=0.1\textheight]{./fig/randomGraphs/que-er}
  % }
  
% \caption{Objective value and runtime for cardinality-constrained maxcut on ca-AstroPh with simulated amounts of content between users.}
%\end{figure}
For these experiments, two random graph models were employed: an Erdős-Rényi (ER) random
graph with $1,000$ nodes and edge probability $p = 1/2$, and a Barabási–Albert (BA) graph 
with $n=10,000$ and $m=m_0=100$.

On the ER graph, results are shown in Figs. \ref{er:obj} and \ref{er:que};
the results on the BA graph are shown in Figs. \ref{ba:obj} and \ref{ba:que}. 
In terms
of cut value, the algorithm of \citet{Gupta2010} performed the best, although the
value produced by \fig was nearly the same. On the ER graph, the next best was
FRG followed by BLITS; whereas on the BA graph, BLITS outperformed FRG in cut value.  
In terms of efficiency of queries, \fig used the smallest number on every evaluated
instance, although the number did increase logarithmically with budget. The number of
queries used by FRG was higher, but after a certain budget remained constant.
The next most efficient was \citet{Gupta2010} followed by BLITS. 
\paragraph{Social Network Monitoring}
For the social network monitoring application, the citation network ca-AstroPh from 
the SNAP dataset collection was used, with $n = 18,772$ users and $198,110$ edges.
Edge weights, which represent the amount of content shared between users, 
were generated uniformly randomly in $[1,10]$. 
The results were similar qualitatively to those for the 
unweighted MaxCut problem presented previously. \fig 
is the most efficient in terms of number of queries, and
\fig is only outperformed in solution quality by \citet{Gupta2010},
which required more than an order of magnitude more queries.
% BLITS performed notably worse than on the random graphs,
% which is likely due to the inaccuracy of using only 30 samples to evaluate expectations:
% as it already required days to complete and is being run as a heuristic, the number
% of samples were not increased.

\paragraph{Effect of Stealing Procedure}
\begin{figure}[t]
  \centering
  \subfigure[][ER instance, $n = 1000$] {
    \label{fig:abl-er}
    \includegraphics[width=0.31\textwidth,height=0.15\textheight]{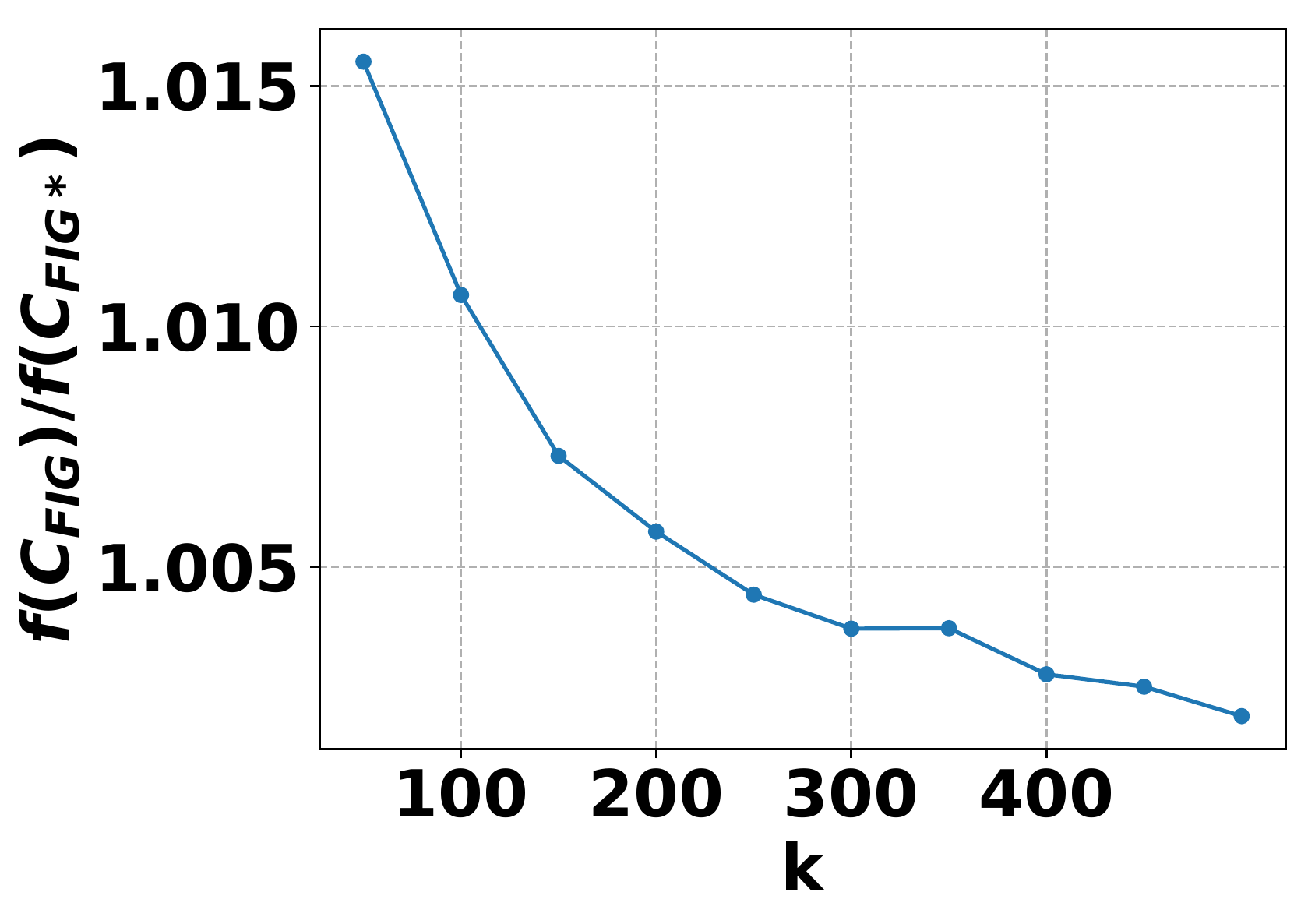}
  }
  \subfigure[][BA instance, $n = 10000$] {
    \label{fig:abl-ba}
    \includegraphics[width=0.31\textwidth,height=0.15\textheight]{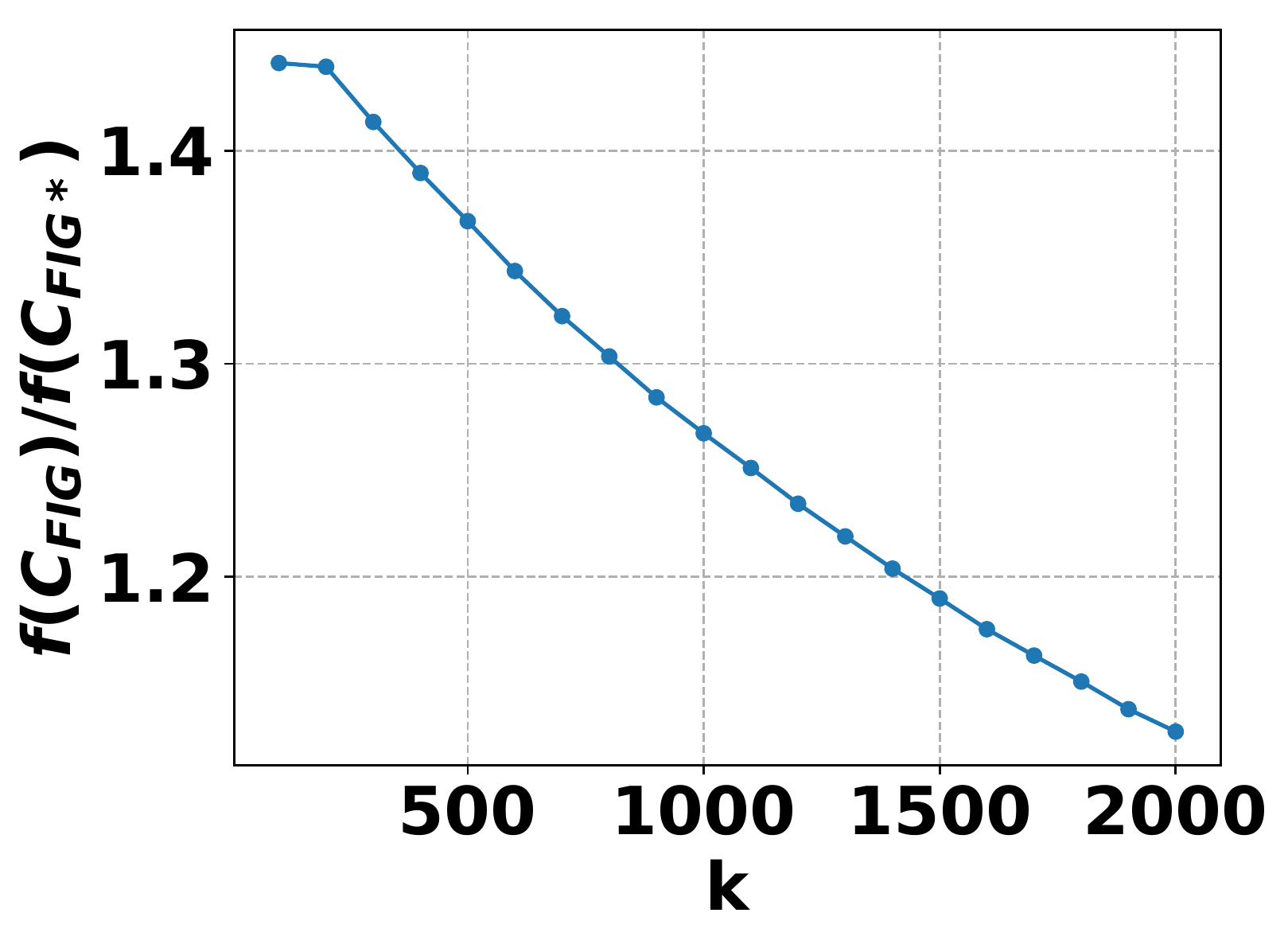}
  }
  \caption{Effect of stealing procedure on solution quality of \fig.}
  \label{fig:ablation} \vspace{-10pt}
\end{figure}
In Fig. \ref{fig:ablation} above, the effect of removing the stealing
    procedure is shown on the random graph instances.
    Let $C_{FIG}$ be the solution returned by FIG, and
    $C_{FIG*}$ be the solution returned by FIG with the 
    stealing procedure removed. Fig. \ref{fig:abl-er}
    shows that on the ER instance, the stealing procedure
    adds at most $1.5\%$ to the solution value; however,
    on the BA instance, Fig. \ref{fig:abl-ba} shows that the
    stealing procedure contributes up to $45\%$ increase
    in solution value, although this effect degrades with larger $k$. 
    This behavior may be explained by the interlaced greedy process 
    being forced to leave good elements out of its solution, which
    are then recovered during the stealing procedure. 
    \clearpage
    \section{Acknowledgements}
    The work of A. Kuhnle was partially supported by
    Florida State University and
    the Informatics Institute of the University of Florida.
    Victoria G. Crawford and the anonymous reviewers
    provided helpful feedback which improved
    the paper.
\bibliographystyle{plainnatfixed}
\bibliography{mend}
\clearpage
\appendix
\section{Proof of Theorem \ref{thm:tg}} \label{apx:proof-thm-tg}
\begin{proof}[Proof of Theorem \ref{thm:tg}]
  \begin{lemma} \label{lemm:minus}
    $$ 4f(C) \ge f\left(O \setminus \{ a_0 \} \right). $$
  \end{lemma}
  \begin{proof}
      Let $A = \argmax_{i \in [k + 1]} f( A_i )$. 
  Let $\hat{O} = O \setminus A_k = \{o_0, \ldots, o_{l-1}\}$ be ordered such that 
  for each $i \in [l]$, $o_i \not \in B_i$; this ordering is
  possible since $B_0 = \emptyset$ and $l \le k$. Also,
  for each $i \in [l]$, let $\hat{O}_i = \{o_0,\ldots,o_{i-1}\}$,  and let $\hat{O}_0 = \emptyset$.
  Then 
  \begin{align*}
    f( O \cup A_k ) - f( A_k ) &= \sum_{i=0}^{l-1} f_{o_i} ( \hat{O}_i \cup A_k )\\
    &\le \sum_{i=0}^{l-1} f_{o_i} (A_i) \\
    &\le \sum_{i=0}^{l-1} f_{a_i} (A_i) = f( A_l ), \\
  \end{align*}
  where the first inequality follows from submodularity, the 
  second inequality follows from the greedy choice $a_i = \argmax_{x \in [n] \setminus (A_i \cup B_i)} f_x( A_{i} )$ and the fact
  that $o_i \not \in B_i$. Hence
  \begin{equation} \label{ineq:A}
    f(O \cup A_k) \le f(A_l) + f(A_k) \le 2f(A).
  \end{equation}

  Let 
  $B = \argmax_{i \in [k + 1]} f( B_i )$.
  Let $\hat{O} = O \setminus \left( \{ a_0 \} \cup B_k \right) = \{o_0, \ldots, o_{l-1}\}$ be ordered such that 
  for each $i \in [l]$, $o_i \not \in A_{i + 1}$; this ordering is
  possible since $A_1 = \{ a_0 \}$, $a_0 \not \in \hat{O}$, and $l \le k$. Also,
  for each $i \in [l]$, let $\hat{O}_i = \{o_0,\ldots,o_{i-1}\}$,
  and let $\hat{O}_0 = \emptyset$.
  Then 
  \begin{align*}
    f( ( O \setminus \{ a_0 \} ) \cup B_k ) - f( B_k ) &= \sum_{i=0}^{l-1} f_{o_i} ( \hat{O}_i \cup B_k )\\
    &\le \sum_{i=0}^{l-1} f_{o_i} (B_i) \\
    &\le \sum_{i=0}^{l-1} f_{b_i} (B_i) = f( B_l ), \\
  \end{align*}
  where the first inequality follows from submodularity, the 
  second inequality follows from the greedy choice $b_i = \argmax_{x \in [n] \setminus (A_{i + 1} \cup B_i)} f_x( B_{i} )$ and the fact
  that $o_i \not \in A_{i + 1}$. Hence
  \begin{equation} \label{ineq:B}
    f( ( O \setminus \{ a_0 \} )\cup B_k) \le f(B_l) + f(B_k) \le 2f(B).
  \end{equation}
  By inequalities (\ref{ineq:A}), (\ref{ineq:B}), the fact that $A_k \cap B_k = \emptyset$, and submodularity, it holds that
  \begin{equation*}
    f( O \setminus \{ a_0 \} ) \le f( O \cup A_k ) + f( ( O \setminus \{ a_0 \} \cup B_k ) \le 2 ( f(A) + f(B) ) \le 4 f(C).  
  \end{equation*}
  \end{proof}
  \begin{lemma} \label{lemm:plus}
     $$ 4f(C) \ge f\left(O \cup \{ a_0 \} \right). $$
  \end{lemma}
  \begin{proof}
    Let $D = \argmax_{i \in [k + 1]} f( A_i )$. 
  Let $\hat{O} = O \setminus D_k = \{o_0, \ldots, o_{l-1}\}$ be ordered such that 
  for each $i \in [l]$, $o_i \not \in E_i$; this ordering is
  possible since $E_0 = \emptyset$ and $l \le k$. Also,
  for each $i \in [l]$, let $\hat{O}_i = \{o_0,\ldots,o_{i-1}\}$,
  and let $\hat{O}_0 = \emptyset$.
  Then 
  \begin{align*}
    f( O \cup D_k ) - f( D_k ) &= \sum_{i=0}^{l-1} f_{o_i} ( \hat{O}_i \cup D_k )\\
    &\le \sum_{i=0}^{l-1} f_{o_i} (D_i) \\
    &\le \sum_{i=0}^{l-1} f_{d_i} (D_i) = f( D_l ), \\
  \end{align*}
  where the first inequality follows from submodularity, the 
  second inequality follows from the greedy choice $d_i = \argmax_{x \in [n] \setminus (D_i \cup E_i)} f_x( D_{i} )$ and the fact
  that $o_i \not \in E_i$. Hence
  \begin{equation} \label{ineq:D}
    f(O \cup D_k) \le f(D_l) + f(D_k) \le 2f(D).
  \end{equation}

  Let 
  $E = \argmax_{i \in [k + 1]} f( E_i )$.
  Let $\hat{O} = O \setminus E_k = \{o_0, \ldots, o_{l-1}\}$ be ordered such that 
  for each $i \in [l]$, $o_i \not \in D_{i + 1}$; this ordering is
  possible since $D_1 = \{ a_0 \}$, $a_0 \not \in \hat{O}$ (since $a_0 \in E_k$), 
  and $l \le k$. Also,
  for each $i \in [l]$, let $\hat{O}_i = \{o_0,\ldots,o_{i-1}\}$,
  and let $\hat{O}_0 = \emptyset$.
  Then 
  \begin{align*}
    f( O \cup E_k ) - f( E_k ) &= \sum_{i=0}^{l-1} f_{o_i} ( \hat{O}_i \cup E_k )\\
    &\le \sum_{i=0}^{l-1} f_{o_i} (E_i) \\
    &\le \sum_{i=0}^{l-1} f_{e_i} (E_i) = f( E_l ), \\
  \end{align*}
  where the first inequality follows from submodularity, the 
  second inequality follows from the greedy choices
  $e_0 = \argmax_{x\in [n]} f(x)$, and if $i > 0$, $e_i = \argmax_{x \in [n] \setminus (D_{i + 1} \cup E_i)} f_x( E_{i} )$ and the fact
  that $o_i \not \in D_{i + 1}$. Hence
  \begin{equation} \label{ineq:E}
    f( ( O \cup E_k) \le f(E_l) + f(E_k) \le 2f(E).
  \end{equation}
  By inequalities (\ref{ineq:D}), (\ref{ineq:E}), the fact that $D_k \cap E_k = \{ a_0 \}$, and submodularity, it holds that
  \begin{equation*}
    f( O \cup \{ a_0 \} ) \le f( O \cup D_k ) + f( ( O  \cup E_k ) \le 2 ( f(D) + f(E) ) \le 4 f(C).  
  \end{equation*}
  \end{proof}
  The proof of the theorem follows from Lemmas \ref{lemm:minus}, \ref{lemm:plus}, and the fact that one of the statements $a_0 \in O$ or $a_0 \not \in O$ must hold;
  hence, either $O \cup \{ a_0 \} = O$ or $O \setminus \{ a_0 \} = O$.
\end{proof}
\section{Proofs for Theorem \ref{thm:ftg}} \label{apx:proof-thm-ftg}
\begin{proof}[Proof of Inequality \ref{ineq:fast-A}]
  Let $A = \{a_0, \ldots, a_{|A| - 1} \}$ be ordered as specified
  by \fig. Likewise, let
  $B = \{b_0, \ldots, b_{|B| - 1} \}$ be ordered as specified
  by \fig.
\begin{lemma} \label{lemma:fast-order-A}
  $O \setminus A = \{o_0,\ldots,o_{l-1} \}$ can be ordered such that
  \begin{equation} \label{ineq:fast-marg-A}
    f_{o_i}(A_i) \le (1 + 2 \delta) f_{a_i}(A_i),
  \end{equation}
  if $i \in [|A|]$.
\end{lemma}
\begin{proof}
  Order $o \in (O \setminus A) \cap B = \{o_0, \ldots, o_{\ell - 1} \} $ by
  the order in which these elements were added into $B$. Order the remaining
  elements of $O \setminus A$ arbitrarily. Then, when 
  $a_i$ was chosen by \add, it holds that $o_i \not \in B_i$.
  Also, it is true $o_i \not \in A_i$;
  hence $o_i$ was not added into some (possibly non-proper) subset $A'_i$ of $A_i$ 
  at the previous threshold value $\frac{\tau_{A_i}}{ (1 - \delta)}$.
  Hence $f_{o_i}(A_i) \le f_{o_i}(A'_i) < \frac{\tau_{A_i}}{(1 - \delta)}$,
  since $o_i \not \in B_i$.
  Since $f_{a_i}(A_i) \ge \tau_{A_i}$ and $\delta < 1/2$, 
  inequality (\ref{ineq:fast-marg-A}) follows.
\end{proof}  
  Order $\hat{O} = O\setminus A = \{o_0, \ldots, o_{l-1} \}$ as 
  indicated in the proof of Lemma \ref{lemma:fast-order-A},
  and let $\hat{O}_i = \{o_0, \ldots, o_{i-1} \}$, if $i \ge 1$,
  $\hat{O}_0 = \emptyset$.
  Then 
  \begin{align*}
    f( O \cup A ) - f( A ) &= \sum_{i=0}^{l-1} f_{o_i} ( \hat{O}_i \cup A )\\
    &= \sum_{i=0}^{|A| - 1} f_{o_i} ( \hat{O}_i \cup A ) + \sum_{i=|A|}^{l - 1} f_{o_i} ( \hat{O}_i \cup A ) \\
    &\le \sum_{i=0}^{|A| - 1} f_{o_i} ( A_i ) + \sum_{i=|A|}^{l - 1} f_{o_i} ( A ) \\
    &\le \sum_{i=0}^{|A| - 1} (1 + 2\delta ) f_{a_i} ( A_i ) + \sum_{i=|A|}^{l - 1} f_{o_i} ( A ) \\
    &\le ( 1 + 2 \delta ) f(A) + \delta M,
  \end{align*}
  where any empty sum is defined to be 0; the first inequality follows by submodularity, the second follows from Lemma \ref{lemma:fast-order-A}, and the third follows from the definition of $A$, and the facts that $\max_{x \in [n] \setminus B_{|A|}} f_x(A) < \stopGain $ and $l - |A| \le k$.
\end{proof}
\begin{proof}[Proof of Inequality \ref{ineq:fast-D}]
  As in the proof of Inequality \ref{ineq:fast-A}, it suffices
  to establish the following lemma.
\end{proof}
  \begin{lemma} 
  $O \setminus D  = \{o_0,\ldots,o_{l-1} \}$ can be ordered such that
  \begin{equation} \label{ineq:fast-marg-D}
    f_{o_i}(D_i) \le (1 + 2 \delta) f_{d_i}(D_i),
  \end{equation}
  for $i \in [|D|]$.
\end{lemma}
\begin{proof}
  Order $o \in (O \setminus D ) \cap E = \{o_0, \ldots, o_{\ell - 1} \} $ by
  the order in which these elements were added into $E$. Order the remaining
  elements of $O \setminus D$ arbitrarily. Then, when 
  $d_i$ was chosen by \add, it holds that $o_i \not \in E_{i}$.
  Also, it is true $o_i \not \in D_i$;
  hence $o_i$ was not added into some (possibly non-proper) subset $D'_i$ of $D_i$ 
  at the previous threshold value $\frac{\tau_{D_i}}{ (1 - \delta)}$.
  Hence $f_{o_i}(D_i) \le f_{o_i}(D'_i) < \frac{\tau_{D_i}}{(1 - \delta)}$,
  since $o_i \not \in E_{i}$.
  Since $f_{d_i}(D_i) \ge \tau_{D_i}$ and $\delta < 1/2$, 
  inequality (\ref{ineq:fast-marg-D}) follows.
\end{proof}
\begin{proof}[Proof of Inequality \ref{ineq:fast-E}]
  As in the proof of Inequality \ref{ineq:fast-A}, it suffices
  to establish the following lemma.
  \begin{lemma} 
  $O \setminus E  = \{o_0,\ldots,o_{l-1} \}$ can be ordered such that
  \begin{equation} \label{ineq:fast-marg-E}
    f_{o_i}(E_i) \le (1 + 2 \delta) f_{e_i}(E_i),
  \end{equation}
  for $i \in [|E|]$.
\end{lemma}
\begin{proof}
  Order $o \in (O \setminus E ) \cap D = \{o_0, \ldots, o_{\ell - 1} \} $ by
  the order in which these elements were added into $D$. Order the remaining
  elements of $O \setminus E$ arbitrarily. Then, when 
  $e_i$ was chosen by \add, it holds that $o_i \not \in D_{i + 1}$, since
  $D_1 = \{ a_0 \}$ and $a_0 = d_0 \not \in O \setminus E$.
  Also, it is true $o_i \not \in E_i$;
  hence $o_i$ was not added into some (possibly non-proper) subset $E'_i$ of $E_i$ 
  at the previous threshold value $\frac{\tau_{E_i}}{ (1 - \delta)}$.
  Hence $f_{o_i}(E_i) \le f_{o_i}(E'_i) < \frac{\tau_{E_i}}{(1 - \delta)}$,
  since $o_i \not \in D_{i+1}$.
  Since $f_{e_i}(E_i) \ge \tau_{E_i}$ and $\delta < 1/2$, 
  inequality (\ref{ineq:fast-marg-E}) follows.
\end{proof}
\end{proof}
\section{Stealing Procedure for FastInterlaceGreedy} \label{apx:steal}
In this section, an $O(k)$ procedure is described,
which may improve the quality of the solution found by
FastInterlaceGreedy (a similar procedure could also be employed for InterlaceGreedy).

Let $A,B,C,D,E$ have their values at the termination of FastInterlaceGreedy.
Then calculate the sets $G = \{ B_c = f(C) - f(C \setminus \{ c \}) : c \in C \}$
and  $H = \{ A_x = f( C \cup \{ x \} ) - f(C) : x \in A \cup B \cup D \cup E \}$.
Then sort $G = (B_{c_1}, \ldots, B_{c_k})$ in non-decreasing order and sort $H = (A_{x_1}, \ldots, A_{x_l})$ in non-increasing order. 
Computing and sorting these sets requires $O(k \log k)$ time (and only $O(k)$ queries to $f$).

Finally, iterate through the elements of $G$ in the sorted order, and 
if $B_{c_i} < A_{x_i}$ then $C$ is assigned $C \setminus \{ c_i \} \cup \{ x_i \}$
if this assignment increases the value $f( C )$.

\section{Proof for Tight Examples}
\label{apx:tight}
\begin{proof}[Proof of Prop. \ref{prop:sm}]
  Submodularity will be verified by checking the inequality
  \begin{equation} \label{submod:apx}
    f(S) + f(T) \ge f(S \cup T) + f (S \cap T)
  \end{equation}
  for all $S, T \subseteq U$. 
  \begin{itemize}
\item \textbf{case $a \in S \cap T$, $b \not \in T \cup S$.}       
      Then Ineq. (\ref{submod:apx}) becomes
      $$ \frac{ |S \cap O| }{ 2k } + \frac{ |T \cap O| }{ 2k } + \frac{2}{k} \ge \frac{|S \cap T \cap O |}{ 2k } + \frac{ |(S \cup T) \cap O| }{ 2k } + \frac{2}{k}, $$
      which holds.
  \item \textbf{case $a \in S \setminus T$, $b \in T \setminus S$.}       
      Then Ineq. (\ref{submod:apx}) becomes
      $$ \frac{ |S \cap O| }{ 2k } + \frac{ |T \cap O| }{ 2k } + \frac{2}{k} \ge \frac{|S \cap T \cap O |}{ k }, $$
      which holds.
    \item \textbf{case $a \in S \setminus T$, $b \in S \setminus T$.}  Then Ineq. (\ref{submod:apx}) becomes
      $$ \frac{ |T \cap O| }{k} \ge \frac{|S \cap T \cap O |}{ k }, $$
      which holds.
    \item \textbf{case $a \in S \setminus T$, $b \in S \cap T$.}  Then Ineq. (\ref{submod:apx}) becomes
      $$ \frac{ |T \cap O| }{ 2k } + \frac{1}{k} \ge \frac{|S \cap T \cap O |}{ 2k } + \frac{1}{k}, $$
      which holds.
    \item \textbf{case $a \in S \cap T$, $b \in S \cap T$.}  Then Ineq. (\ref{submod:apx}) becomes
      $$ 0 \ge 0, $$
      which holds
    \item \textbf{case $a \not \in S \cup T$, $b \not \in S \cup T$.}  Then Ineq. (\ref{submod:apx}) becomes
      $$ | S \cap O | + |T \cap O| \ge | (S \cup T) \cap O | + | (S \cap T) \cap O |,$$
      which holds.
    \item \textbf{case $a \in S \setminus T$, $b \not \in S \cup T$.}  Then Ineq. (\ref{submod:apx}) becomes
      $$ \frac{| S \cap O |}{2k} + \frac{1}{k} + \frac{|T \cap O|}{k} \ge \frac{| (S \cup T) \cap O |}{2k} + \frac{1}{k} + \frac{| (S \cap T) \cap O |}{k},$$
      which holds.
  \end{itemize}
  The remaining cases follow symmetrically.
\end{proof}

%%% Local Variables:
%%% mode: latex
%%% TeX-master: "nips.tex"
%%% End:

\end{document}